\documentclass[aps,pre,twocolumn,superscriptaddress,runinaffiliation,showkeys,showpacs]{revtex4-1}

\usepackage{graphicx}
\usepackage{epsfig}
\usepackage{subfigure}
\usepackage{amsmath,amsthm,amssymb,mathtools,mathrsfs}

\usepackage{times}
\usepackage{microtype}

\newtheorem{mydef}{Definition}
\newtheorem{mythm}{Theorem}
\newtheorem{mylma}{Lemma}

\begin{document}

\title{Nonnegative Decomposition of Multivariate Information}

\author{Paul L. Williams}
\email{plw@indiana.edu}
\affiliation{Cognitive Science Program and}
\author{Randall D. Beer}
\affiliation{Cognitive Science Program and}
 \affiliation{School of Informatics and Computing\\ Indiana University, Bloomington, Indiana 47406 USA}

\date{\today}

\begin{abstract}
Of the various attempts to generalize information theory to multiple variables, the most widely utilized, interaction information, suffers from the problem that it is sometimes negative. Here we reconsider from first principles the general structure of the information that a set of sources provides about a given variable. We begin with a new definition of redundancy as the minimum information that any source provides about each possible outcome of the variable, averaged over all possible outcomes. We then show how this measure of redundancy induces a lattice over sets of sources that clarifies the general structure of multivariate information.  Finally, we use this redundancy lattice to propose a definition of partial information atoms that exhaustively decompose the Shannon information in a multivariate system in terms of the redundancy between synergies of subsets of the sources. Unlike interaction information, the atoms of our partial information decomposition are never negative and always support a clear interpretation as informational quantities. Our analysis also demonstrates how the negativity of interaction information can be explained by its confounding of redundancy and synergy.
\end{abstract}


\maketitle
 
 \section{Introduction}
 
From its roots in Shannon's seminal work on reliability and coding in communication systems, information theory has grown into a ubiquitous general tool for the analysis of complex systems, with application in neuroscience, genetics, physics, machine learning, and many other areas. Somewhat surprisingly, the vast majority of work in information theory concerns only the simplest possible case: the information that a single variable provides about another.  This is quantified by Shannon's mutual information, which is by far the most widely used concept from information theory \cite{Shannon1949}.  The second most popular concept, conditional mutual information, considers interactions between multiple variables in only the most rudimentary sense: it seeks to eliminate the influence of other variables in order to isolate the dependency between two variables of interest.  In contrast, many of the most interesting and challenging scientific questions, such as many-body problems in physics \cite{Pines1997}, $n$-person games in game theory \cite{Luce1989}, and population coding in neuroscience \cite{Dayan2001, Rieke1999}, involve understanding the structure of interactions between three or more variables.

The two main attempts to generalize information theory to multivariate interactions are the \emph{total correlation} proposed by Watanabe \cite{Watanabe1960} (also known as the multivariate constraint \cite{Garner1962}, multiinformation \cite{Studeny1998}, and integration \cite{Tononi1994}) and the \emph{interaction information} of McGill \cite{McGill1954} (also known as multiple mutual information \cite{Han1980}, co-information \cite{Bell2003}, and synergy \cite{Gawne1993}).  The total correlation, as its name suggests, measures the total amount of dependency between a set of variables as a single monolithic quantity.  Thus, the total correlation does not provide any insight into how dependencies are distributed amongst the variables, i.e., it says nothing about the \emph{structure} of multivariate information.

In contrast, interaction information was proposed as a measure of the amount of information bound up in a set of variables beyond that which is present in any subset of those variables.  Thus, entropy and mutual information correspond to first- and second-order interaction information, respectively, and together with its third-, fourth-, and higher-order variants, interaction information provides a way of characterizing the structure of multivariate information.  Interaction information is also the natural generalization of mutual information when Shannon entropy is viewed as a signed measure on information diagrams \cite{Bell2003,Yeung2008,Cover2006}.  However, the wider use of interaction information has largely been hampered by the ``odd'' \cite{Bell2003} and ``unfortunate'' \cite{Cover2006} property that, for three or more variables, the interaction information can be negative (see also \cite{Han1980,Yeung2008,Takano1974, Tsujishita1995, Zhang1998}). For information as it is commonly understood, it is entirely unclear what it means for one variable to provide ``negative information'' about another.  Moreover, as we demonstrate below, the confusing property of negativity is actually symptomatic of deeper problems regarding the interpretation of interaction information for larger systems.  As a result, there remains no generally accepted extension of information theory for characterizing the structure of multivariate interactions.  

Here we formulate a new perspective on the structure of multivariate information.  Beginning from first principles, we consider the general structure of the information that a set of sources provide about a given variable.  We propose a new definition of redundancy as the minimum information that any source provides about each outcome of the variable, averaged over all possible outcomes.  Then we show how this definition can be used to exhaustively decompose the Shannon information in a multivariate system into partial information atoms consisting of redundancies between synergies of subsets of the sources.  We also demonstrate that partial information forms a lattice that clarifies the general structure of multivariate information.  Unlike interaction information, the atoms of our partial information decomposition are never negative and always support a clear interpretation as informational quantities.  Finally, our analysis also demonstrates how the negativity of interaction information can be explained by its confounding of redundant and synergistic interactions.  

\section{The Structure of Multivariate Information}

Suppose we are given a random variable $S$ and a random vector ${\bf R} = \{R_1, R_2, \ldots, R_{n-1}\}$.  Then our goal is to decompose the information that ${\bf R}$ provides about $S$ in terms of the partial information contributed either individually or jointly by various subsets of ${\bf R}$.  For example, in a neuroscience context, $S$ may correspond to a stimulus that takes on different values and ${\bf R}$ to the evoked responses of different neurons.  In this case, we would like to quantify the information that the joint neural response provides about the stimulus, and to distinguish between information due to responses of individual neurons versus combinations of them \cite{Rieke1999,Gawne1993}.

\begin{figure}[t!]
  \centering
    \includegraphics{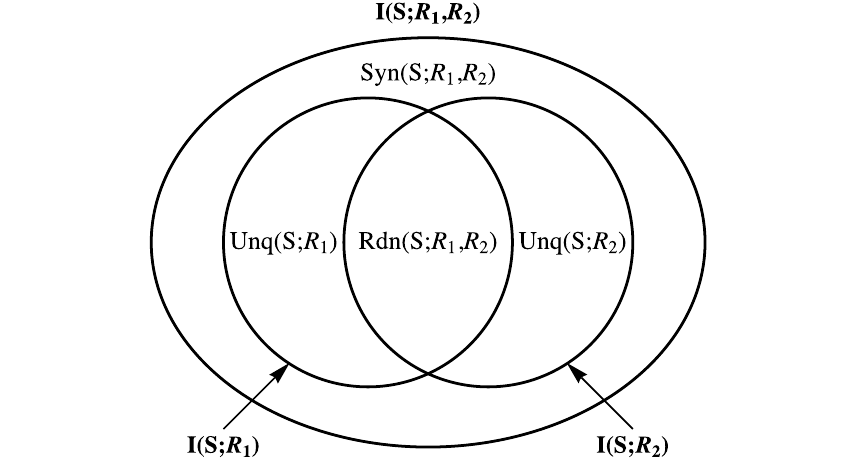}
\caption{Structure of multivariate information for 3 variables.  Labelled regions correspond to unique information (\emph{Unq}), redundancy (\emph{Rdn}), and synergy (\emph{Syn}).}
\label{generic_syn_rdn}
\end{figure}

Consider the simplest case of a system with three variables.  How much total information does ${\bf R} = \{R_1, R_2\}$ provide about $S$? How do $R_1$ and $R_2$ contribute to the total information?  The answer to the first question is given by the mutual information $I(S; R_1, R_2)$, while for the latter we can identify three distinct possibilities.  First, $R_1$  may provide information that $R_2$ does not, or vice versa (\emph{unique information}).  For example, if $R_1$ is a copy of $S$ and $R_2$ is a degenerate random variable, then the total information from ${\bf R}$  reduces to the unique information from  $R_1$. Second, $R_1$ and $R_2$ may provide the same or overlapping information (\emph{redundancy}).  For example, if $R_1$ and $R_2$ are both copies of $S$ then they redundantly provide complete information.  Third, the combination of $R_1$ and $R_2$ may provide information that is not available from either alone (\emph{synergy}).  A well-known example for binary variables is the exclusive-OR function $S = R_1 \oplus R_2$, in which case $R_1$ and $R_2$ individually provide no information but together provide complete information.  Thus, intuitively, the total information from ${\bf R}$ decomposes into unique information from $R_1$ and $R_2$, redundant information shared by $R_1$ and $R_2$, and synergistic information contributed jointly by $R_1$ and $R_2$ (FIG. 1).  

In sum, for three variables we can identify unique information, redundancy, and synergy as the basic atoms of multivariate information.  In fact, as later developments will clarify, unique information is best thought of as a degenerate form of redundancy or synergy, so that redundancy and synergy alone constitute the basic building blocks of multivariate information.  In particular, we will find that various combinations of redundancy and synergy, which may at first sound paradoxical, play a fundamental role in structuring multivariate information in higher dimensions.  Next we proceed to formalize these ideas, beginning with the problem of defining a measure of redundancy.  

\section{Measuring Redundancy}

Let ${\bf A}_1, {\bf A}_2, \ldots, {\bf A}_k$ be nonempty and potentially overlapping subsets of ${\bf R}$, which we call \emph{sources}.  How can we quantify the redundant information that all sources provide about $S$?

Of course, the information supplied by each ${\bf A}_i$  is given simply by $I(S; {\bf A}_i)$, the mutual information between $S$ and ${\bf A}_i$.  However, it is crucial to note that mutual information is actually a measure of \emph{average} or \emph{expected} information, where the expected value is taken over outcomes of the random variables.  Thus, for instance, two sources might provide the same average amount of information, while also providing information about different outcomes of $S$.  Stated formally, the information provided by a source ${\bf A}$ can be written as
\begin{equation}
I(S; {\bf A}) = \sum_{s} p(s) I(S=s; {\bf A})
\end{equation} 
where the \emph{specific information} $I(S=s; {\bf A})$ quantifies the information associated with a particular outcome $s$ of $S$.  Various definitions of specific information have been proposed to quantify different relationships between $S$ and ${\bf A}$ (see Appendix A), but for our purposes the most useful is
\begin{equation}\label{specInfo}
I(S=s; {\bf A}) = \sum_{{\bf a}} p({\bf a}|s) \bigg[ \log \frac{1}{p(s)} - \log \frac{1}{p(s|{\bf a})} \bigg].
\end{equation}
The term $\frac{1}{p(s)}$ is called the surprise of $s$, so $I(S=s; {\bf A})$ is the average reduction in surprise of $s$ given knowledge of ${\bf A}$. In other words, $I(S=s; {\bf A})$ quantifies the information that ${\bf A}$ provides about each particular outcome $s \in S$, while $I(S; {\bf A})$ is the expected value of this quantity over all outcomes of $S$.

Given these considerations, a natural measure of redundancy is the expected value of the minimum information that any source provides about each outcome of $S$, or
\begin{equation}\label{IminDef}
I_{\min}(S; \{{\bf A}_1, {\bf A}_2, \ldots, {\bf A}_k\}) = \sum_s p(s) \min_{{\bf A}_i} I(S = s; {\bf A}_i).
\end{equation}
$I_{\min}$ captures the idea that redundancy is the information common to all sources (the minimum information that any source provides), while taking into account that sources may provide information about different outcomes of $S$.  Note that, like the mutual information, $I_{\min}$ is also an expected value of specific information terms.

 $I_{\min}$ also has several important properties that further support its interpretation as a measure of redundancy.  First, $I_{\min}$ is nonnegative, a property that follows directly from the nonnegativity of specific information (see Appendix D).  Second,  $I_{\min}$ is less than or equal to $I(S; {\bf A}_i)$ for all ${\bf A}_i$'s, with equality if and only if $I(S=s; {\bf A}_i) = I(S=s; {\bf A}_j)$ for all $i$ and $j$ and all $s \in S$.  Thus, as one would hope, the amount of redundant information is bounded by the information provided by each source, with equality if and only if all sources provide the exact same information about $S$.  Finally, and closely related to the previous property, for a given source ${\bf A}$ the amount of information redundant with ${\bf A}$ is maximal for $I_{\min}(S; \{{\bf A}\}) = I(S;{\bf A})$.  In other words, redundant information is maximized by the ``self-redundancy,'' analogous to the property that mutual information is maximized by the self-information $I(S; S) = H(S)$.

What are the distinct ways in which collections of sources might contribute redundant information? Formally, answering this question means identifying the domain of $I_{\min}$.  Thus far, we have assumed that the natural domain is the collection of all possible sets of sources, but in fact this can be greatly simplified.  To illustrate, consider two sources, ${\bf A}$ and ${\bf B}$, with ${\bf A}$ a subset of ${\bf B}$.  Clearly, any information provided by ${\bf A}$ is also provided by ${\bf B}$, so the redundancy between ${\bf A}$ and ${\bf B}$ reduces to the self-redundancy for ${\bf A}$,
\begin{equation*}
I_{\min}(S; \{{\bf A}, {\bf B}\}) = I_{\min}(S; \{{\bf A}\}) = I(S; {\bf A}).
\end{equation*}
Furthermore, for any source ${\bf C}$, it follows that $I_{\min}(S; \{{\bf A}, {\bf B}, {\bf C}\}) = I_{\min}(S; \{{\bf A}, {\bf C}\})$.  Extending this idea, for any collection of sources where some are supersets of others, the redundancy for that collection is equivalent to the redundancy with all supersets removed.  Thus, the domain for $I_{\min}$ can be reduced to the collection of all sets of sources such that no source is a superset of any other.  Formally, this set can be written as 
\begin{equation}
\mathcal{A}({\bf R}) = \{\alpha \in \mathcal P_{1}(\mathcal P_{1}({\bf R})) : \forall {\bf A}_i, {\bf A}_j \in \alpha, {\bf A}_i \not\subset {\bf A}_j\}, 
\end{equation}
where $\mathcal P_1( {\bf R} ) = \mathcal P( {\bf R} ) \setminus \{\emptyset\}$ is the set of all nonempty subsets of ${\bf R}$.  Henceforth, we will denote elements of $\mathcal{A}({\bf R})$, corresponding to collections of sources, with bracketed expressions containing only the indices for each source.  For instance, $\{\{R_1, R_2\}\}$ will be $\{12\}$, $\{\{R_1\}$, $\{R_2, R_3\}\}$ will be $\{1\}\{23\}$, and so forth.  

The possibilities for redundancy are also naturally structured, which is shown by extending the same line of reasoning to define an ordering $\preccurlyeq$ on the elements of $\mathcal{A}({\bf R})$.  Consider two collections of sources, $\alpha, \beta \in \mathcal{A}({\bf R})$, where for each source ${\bf B} \in \beta$ there exists a source ${\bf A} \in \alpha$ with ${\bf A}$ a subset of ${\bf B}$.  This means that for each source ${\bf B} \in \beta$ there is a source ${\bf A} \in \alpha$ such that ${\bf A}$ provides no more information than ${\bf B}$.  The redundant information shared by all ${\bf B} \in \beta$ must therefore at least include any redundant information shared by all ${\bf A} \in \alpha$.  Thus, we can define a partial order over the elements of $\mathcal{A}({\bf R})$ such that one element (collection of sources) is considered to precede another if and only if the latter provides any redundant information that the former provides.  The ordering relation  $\preccurlyeq$ is formally defined as
\begin{equation}
\forall \alpha, \beta \in \mathcal{A}({\bf R}), (\alpha \preccurlyeq \beta \Leftrightarrow \forall {\bf B} \in \beta, \exists {\bf A} \in \alpha, {\bf A} \subseteq {\bf B}).
\end{equation}
Applying this ordering to the elements of $\mathcal{A}({\bf R})$ produces a \emph{redundancy lattice}, in which a higher element provides at least as much redundant information as a lower one (FIG. 2; see Appendix C).

The redundancy lattice provides a wealth of insight into the structure of redundancy.  For instance, from the redundancy lattice it is possible to read off some of the properties of $I_{\min}$ noted earlier.  The property that redundancy for a source is maximized by the self-redundancy can be seen from the fact that any node corresponding to an individual source appears higher in the redundancy lattice than any other node involving that source.  For example, in FIG. 2B, the node labeled \{12\}, corresponding to the self-redundancy for the source $\{R_1, R_2\}$, occurs higher than nodes labeled \{12\}\{13\}, \{12\}\{13\}\{23\}, and \{3\}\{12\}.  Another property of $I_{\min}$ that can be seen from these diagrams relates to the top and bottom elements of the lattice.  The top element corresponds to the self-redundancy for ${\bf R}$, reflecting the fact that $I_{\min}$ is bounded from above by the total amount of information provided by ${\bf R}$.  At the other end of the spectrum, the bottom element corresponds to the redundant information that each individual element of ${\bf R}$ provides, with all other possibilities for redundancy falling between these two extremes.

\section{Partial Information Decomposition}

The redundant information associated with each node of the redundancy lattice includes, but is not limited to, the redundant information provided by all nodes lower in the lattice.  Thus, moving from node to node up the lattice, $I_{\min}$ can be thought of as a kind of ``cumulative information function,'' effectively integrating the information provided by increasingly inclusive collections of sources.  Next, we derive an inverse of $I_{\min}$ called the partial information function (PI-function).  Whereas  $I_{\min}$ quantifies cumulative information, the PI-function measures the partial information contributed uniquely by each particular collection of sources. This partial information will form the atoms into which we decompose the total information that ${\bf R}$ provides about $S$.

\begin{figure}[t!]
  \centering
    \includegraphics{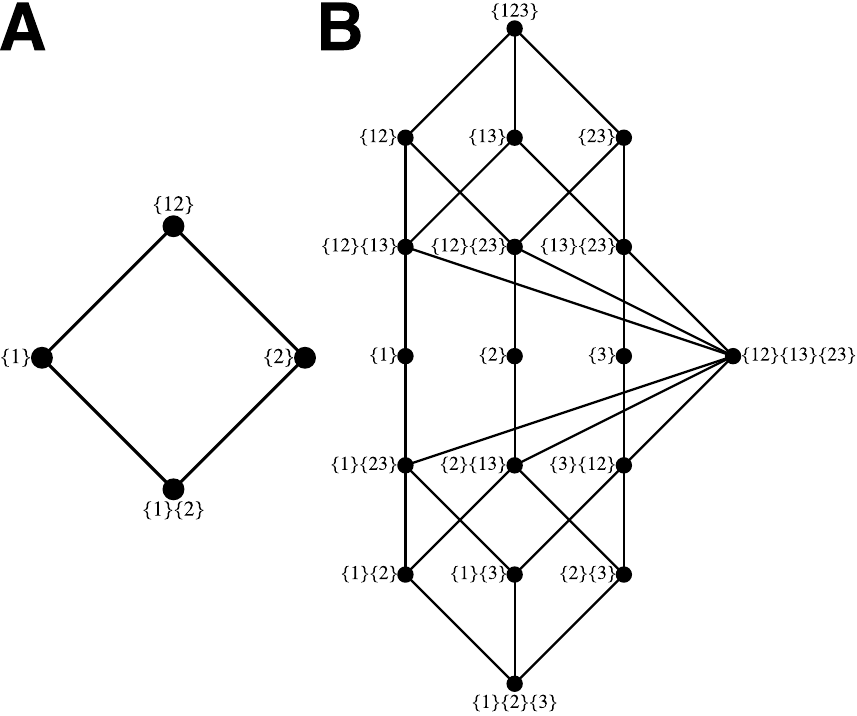}
\caption{Redundancy lattice for (A) 3 and (B) 4 variables.}
\label{redundancy_lattices}
\end{figure}

For a collection of sources $\alpha \in \mathcal{A}({\bf R})$, the PI-function, denoted $\Pi_{{\bf R}}$, is defined implicitly by
\begin{equation}\label{piImplicit}
I_{\min}(S; \alpha) = \sum_{\beta \preccurlyeq \alpha} \Pi_{{\bf R}}(S; \beta).
\end{equation} 
Formally, $\Pi_{{\bf R}}$ corresponds to the the M$\ddot{\text{o}}$bius inverse of $I_{\min}$ \cite{Rota1964,Stanley1997}.  From this relationship, it is clear that $\Pi_{{\bf R}}$ can be calculated recursively as
\begin{equation}\label{piRecursive}
\Pi_{{\bf R}}(S; \alpha) = I_{\min}(S; \alpha) - \sum_{\beta \prec \alpha} \Pi_{{\bf R}}(S; \beta).
\end{equation}
Put into words, $\Pi_{{\bf R}}(S; \alpha)$ quantifies the information provided redundantly by the sources of $\alpha$ that is not provided by any simpler collection of sources  (i.e., any $\beta$  lower than $\alpha$ on the redundancy lattice).  In Appendix D, it is shown that $\Pi_{{\bf R}}$ can be written in closed form as
\begin{equation}\label{piCF}
\Pi_{{\bf R}}(S; \alpha) = I_{\min}(S; \alpha) -  \sum_{s} p(s) \max_{\beta \in \alpha^-} \min_{{\bf B} \in \beta} I(S=s; {\bf B})
\end{equation}
where $\alpha^-$ represents the nodes immediately below $\alpha$ in the redundancy lattice.  From this formulation, it is readily shown that $\Pi_{{\bf R}}$ is nonnegative (see Appendix D), and thus can be naturally interpreted as an informational quantity associated with the sources of $\alpha$. 

The decomposition of mutual information into a sum of PI-terms follows from
\begin{equation}\label{uncondDecomp}
I(S; {\bf A}) = I_{\min}(S; \{{\bf A}\}) = \sum_{\beta \preccurlyeq \{{\bf A}\}} \Pi_{{\bf R}}(S; \beta).
\end{equation}
For the 3-variable case ${\bf R} = \{R_1, R_2\}$, Equation \eqref{uncondDecomp} yields
\begin{align}
I(S; R_1) & = \Pi_{{\bf R}}(S; \{1\}) + \Pi_{{\bf R}}(S; \{1\}\{2\}) \\
& \mbox{and} \notag\\
I(S; R_1, R_2) & = \Pi_{{\bf R}}(S; \{1\}) + \Pi_{{\bf R}}(S; \{2\}) \notag \\
				& + \Pi_{{\bf R}}(S; \{1\}\{2\}) + \Pi_{{\bf R}}(S; \{12\}).
\end{align}

The relationship between these equations can be represented as a \emph{partial information (PI) diagram} (FIG. 3A), which illustrates the way in which the total information that ${\bf R}$ provides about $S$ is distributed across various combinations of sources. Furthermore, comparing this diagram with FIG. 1 makes immediately clear the meaning of each partial information term.  First, from Equation \eqref{piCF}, we have that $\Pi_{{\bf R}}(S; \{1\}\{2\}) = I_{\min}(S; \{1\}\{2\})$, which, from the definition of $I_{\min}$, corresponds to the redundancy for $R_1$ and $R_2$.  The unique information for $R_1$ is given by $\Pi_{{\bf R}}(S; \{1\}) = I(S; R_1) - I_{\min}(S; \{1\}\{2\})$, which is the total information from $R_1$ minus the redundancy, and likewise for $R_2$.  Finally, the additional information provided by the combination of $R_1$ and $R_2$ is given by $\Pi_{{\bf R}}(S; \{12\})$, corresponding to their synergy.

To fix ideas, consider the example in FIG. 4A.  From the symmetry of the distribution, it is clear that $R_1$ and $R_2$ must provide the same amount of information about $S$.  Indeed, this is easily verified, with $I(S; R_1) = I(S; R_2) = -\frac{1}{3} \log \frac{1}{3} - \frac{2}{3} \log \frac{2}{3}$.  However, it is also clear that $R_1$ and $R_2$ provide information about different outcomes of $S$.  In particular, given knowledge of $R_1$, one can determine conclusively whether or not outcome $S=2$ occurs (which is not the case for $R_2$), and likewise for $R_2$ and outcome $S=1$.  This feature is captured by $\Pi_{{\bf R}}(S; \{1\}) = \Pi_{{\bf R}}(S; \{2\}) = \frac{1}{3}$, indicating that $R_1$ and $R_2$ each provide $\frac{1}{3}$ bits of unique information about $S$.  The redundant information, $\Pi_{{\bf R}}(S; \{1\}\{2\}) = \log 3 - \log 2$, captures the fact that knowledge of either $R_1$ or $R_2$ reduces uncertainty about $S$ from three equally likely outcomes to two.  Finally, $R_1$ and $R_2$ also provide $\frac{1}{3}$ bits of synergistic information, i.e., $\Pi_{{\bf R}}(S; \{12\}) = \frac{1}{3}$.  This value reflects the fact that $R_1$ and $R_2$ together uniquely determine whether or not outcome $S=0$ occurs, which is not true for $R_1$ or $R_2$ alone.

Note that, unlike mutual information or interaction information, partial information is \emph{not} symmetric. For instance, the synergistic information that $R_1$ and $R_2$ provide about $S$ is not in general equal to the synergistic information that $S$ and $R_2$ provide about $R_1$.  This property is also illustrated by the example in FIG. 4A.   Given knowledge of $S$, one can uniquely determine the outcome of $R_1$ (and $R_2$), so that $S$ provides complete information about both.  Thus, it is not possible for the combination of $S$ and $R_2$ to provide any additional synergistic information about $R_1$, since there is no remaining uncertainty about $R_1$ when $S$ is known.  In contrast, as was just noted, $R_1$ and $R_2$ provide $\frac{1}{3}$ bits of synergistic information about $S$.  This asymmetry accounts for our decision to focus on information \emph{about} a particular variable $S$ throughout, since in general the analysis will differ depending on the variable of interest.  Note that total information is also asymmetric in this sense, i.e., in general $I(S; R_1, R_2) \neq I(R_1; S, R_2)$ (though, of course, it is symmetric in the sense that $I(S; R_1, R_2) = I(R_1, R_2; S)$). 

The general structure of PI-diagrams becomes clear when we consider the decomposition for four variables (FIG. 3B).  First, note that all of the possibilities for three variables are again present for four.  In particular, each element of ${\bf R}$ can provide unique information (regions labeled \{1\}, \{2\}, and \{3\}), information redundantly with one other variable (\{1\}\{2\}, \{1\}\{3\}, and \{2\}\{3\}), or information synergistically with one other variable (\{12\}, \{13\}, and \{23\}).  Additionally, information can be provided redundantly by all three variables (\{1\}\{2\}\{3\}) or provided by their three-way synergy (\{123\}).  More interesting are the new kinds of terms representing combinations of redundancy and synergy.  For instance, the regions marked \{1\}\{23\}, \{2\}\{13\}, and \{3\}\{12\} represent information that is available redundantly from either one variable considered individually or the other two considered together.  Or, for instance, the region labeled \{12\}\{13\}\{23\} represents the information provided redundantly by the three possible two-way synergies.  In general, the PI-atom for a collection of sources corresponds to the information provided redundantly by the synergies of all sources in the collection.  This point also clarifies our earlier claim that unique information is best thought of as a degenerate case: unique information corresponds to the combination of first-order redundancy and first-order synergy.

In general, a PI-diagram for $n$ variables, $S$ and ${\bf R} = \{R_1, R_2, \ldots, R_{n-1}\}$, consists of the following (see Fig. S2 in Appendix E).  First, for each element $R_i \in {\bf R}$  there is a region corresponding to $I(S; R_i)$.  Then, for every subset ${\bf A}$ of ${\bf R}$ with two or more elements, $I(S; {\bf A})$ is depicted as a region containing $I(S; A)$ for all $A \in {\bf A}$ but not coextensive with $\bigcup_{A \in {\bf A}} I(S; A)$.  The difference between $I(S; {\bf A})$ and $\bigcup_{A \in {\bf A}} I(S; A)$ represents the synergy for ${\bf A}$, the information gained from the combined knowledge of all elements in ${\bf A}$ that is not available from any subset.  In addition, regions of the diagram intersect generically, representing all possibilities for redundancy.  In total, a PI-diagram is composed of the $(n-1)$-th Dedekind number \cite{Comtet1974} of PI-atoms, same as the cardinality of $\mathcal{A}({\bf R})$ (see Appendix C).  As described above, each PI-atom represents the redundancy of synergies for a particular collection of sources, corresponding to one distinct way for the components of ${\bf R}$ to contribute information about $S$.

\begin{figure}[t!]
  \centering
    \includegraphics{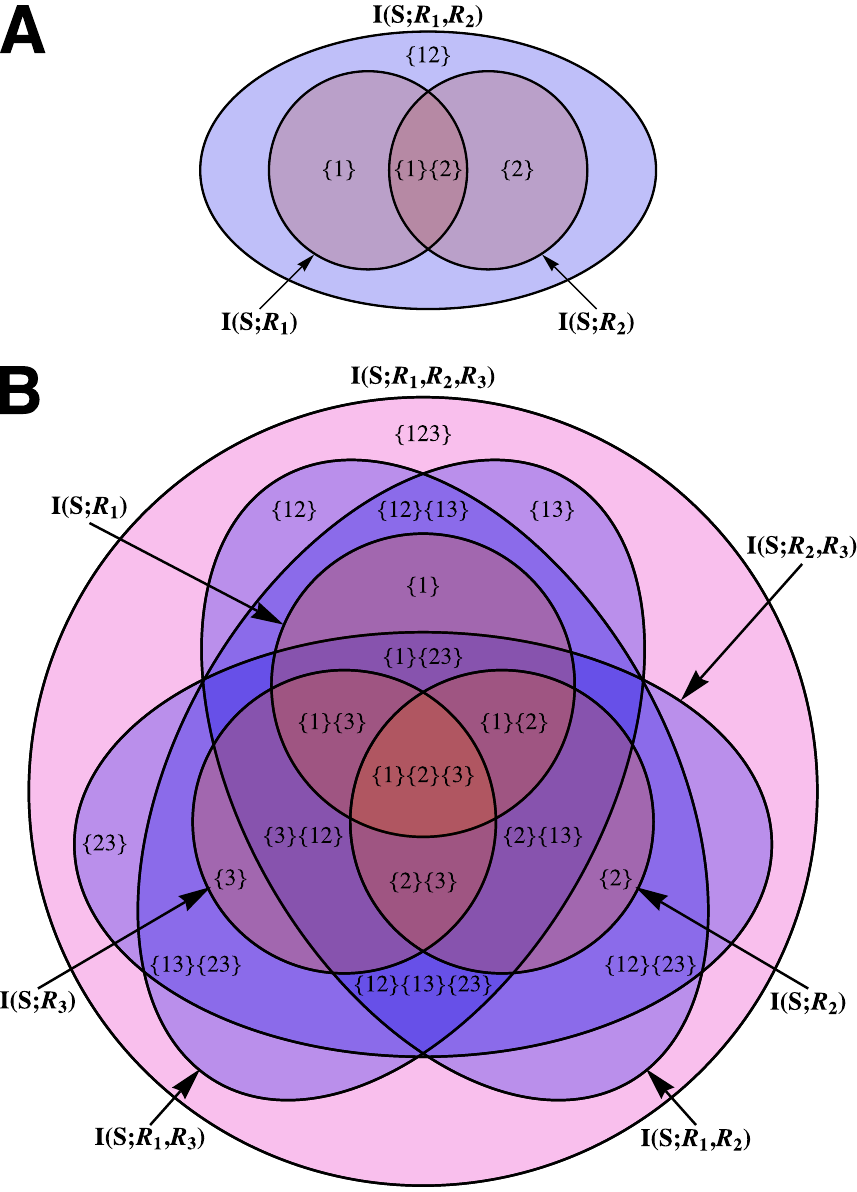}
\caption{Partial information diagrams for (A) 3 and (B) 4 variables.}
\label{pi_diagrams}
\end{figure}

Finally, it is instructive to consider the relationship between the redundancy lattice and PI-diagram for $n$ variables.  First, we note that $I_{\min}$ is analogous to set intersection for PI-diagrams, consistent with the idea of redundancy as overlapping information.  Specifically, $I_{\min}(S; \{{\bf A}_1, {\bf A}_2, \ldots, {\bf A}_k\})$ corresponds to the region $\bigcap_{i} I(S; {\bf A}_i)$.  From this correspondence between $I_{\min}$ and set intersection, we can establish the following connection: for $\alpha, \beta \in \mathcal{A}({\bf R})$, $\alpha$ is lower than $\beta$ in the redundancy lattice if and only if $\bigcap_{{\bf A} \in \alpha} I(S; {\bf A})$ is a subset of $\bigcap_{{\bf B} \in \beta} I(S; {\bf B})$ in the PI-diagram.  Consequently, the redundancy lattice and PI-diagram can be viewed as complementary representations of the same structure, with the PI-diagram a collapsed version of the redundancy lattice formed by embedding regions according to the lattice ordering.  

\section{Why Interaction Information is \\Sometimes Negative}

We next show how PI-decomposition can be used to understand the conditions under which interaction information, the standard generalization of mutual information to multivariate interactions, is negative. The interaction information for three variables is given by
\begin{equation}\label{3varII}
I(S; R_1; R_2) = I(S; R_1 | R_2) - I(S; R_1)
\end{equation}
and for $n > 3$ variables is defined recursively as
\begin{align}\label{nvarII}
I(S; R_1; R_2; \ldots; R_{n-1}) = & I(S; R_1; R_2; \ldots; R_{n-2} | R_{n-1}) \notag \\ 
						& - I(S; R_1; R_2; \ldots; R_{n-2})
\end{align}
where the conditional interaction information is defined by simply including the conditioning in all terms of the original definition.  Interaction information is symmetric for all permutations of its arguments, and is traditionally interpreted as the information shared by all $n$ variables beyond that which is shared by any subset of those variables.    

\begin{figure}[t!]
  \centering
    \includegraphics{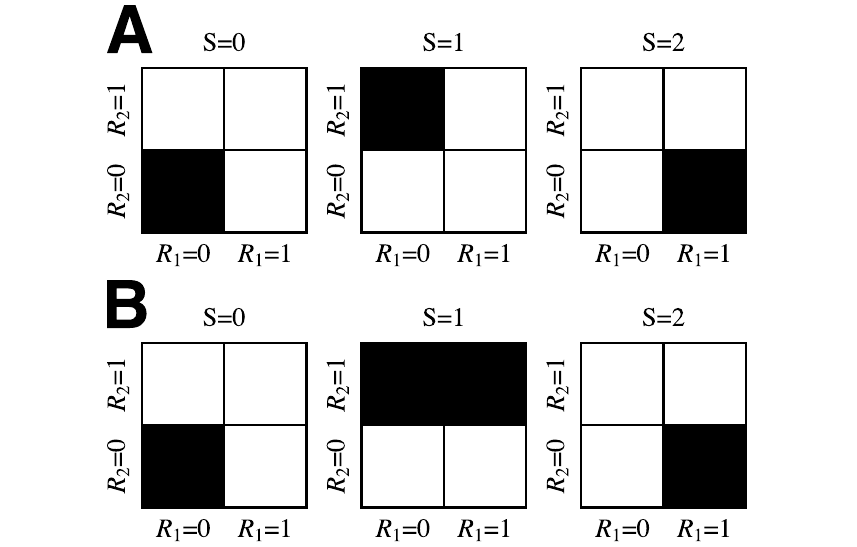}
\caption{Probability distributions for $S \in \{0,1,2\}$ and $R_1, R_2 \in \{0,1\}$.  Black tiles represent equiprobable outcomes. White tiles are zero-probability outcomes.}
\label{sample_dists}
\end{figure}

For 3-variable interaction information, a positive value is naturally interpreted as indicating a situation in which any one variable of the system enhances the correlation between the other two.  For example, a positive value for Equation \eqref{3varII} indicates that knowledge of $R_2$ enhances the correlation between $S$ and $R_1$ (and likewise for all other variable permutations).  Thus, in the terminology used here, a positive value for $I(S; R_1; R_2)$ indicates the presence of synergy.  On the other hand, a negative value for $I(S; R_1; R_2)$ indicates a situation in which any one variable accounts for or ``explains away'' \cite{Pearl1988} the correlation between the other two.  In other words, a negative value for $I(S; R_1; R_2)$ indicates redundancy.  Indeed, $I(S; R_1; R_2)$ is a widely used measure of synergy in neuroscience, where it is interpreted in exactly this way \cite{Brenner2000,Panzeri1999,Schneidman2003,Latham2005}.

The PI-decomposition for 3-variable interaction information (FIG. 5A; see also Fig. S3 in Appendix E) confirms this interpretation, with $I(S; R_1; R_2)$ equal to the difference between the synergistic and the redundant information, i.e.,
\begin{equation}
I(S; R_1; R_2) = \Pi_{{\bf R}}(S; \{12\}) - \Pi_{{\bf R}}(S; \{1\}\{2\}).
\end{equation}
Thus, it is indeed the case that positive values indicate synergy and negative values indicate redundancy.  

However, PI-decomposition also makes clear that $I(S; R_1; R_2)$ confounds redundancy and synergy, with the meaning of interaction information ambiguous for any system that exhibits a mixture of the two  (cf. \cite{Jakulin2003}, who suggest the possibility of mixed redundancy and synergy, but without attempting to disentangle them).  For instance, consider again the example in FIG. 4A.  As described earlier, in this case $R_1$ and $R_2$ provide $\log 3 - \log 2$ bits of redundant information and $\frac{1}{3}$ bits of synergistic information.  Consequently, $I(S; R_1; R_2)$ is negative because there is more redundancy than synergy, despite the fact that the system clearly exhibits synergistic interactions.  As a second example, consider the distribution in FIG. 4B.  In this case, $R_1$ and $R_2$ provide $\frac{1}{2}$ bits of redundant information, corresponding to the fact that knowledge of either $R_1$ or $R_2$ reduces uncertainty about the outcomes $S=0$ and $S=2$.  Additionally, $R_1$ and $R_2$ provide $\frac{1}{2}$ bits of synergistic information, reflecting the fact that $R_1$ and $R_2$ together provide complete information about outcomes $S=0$ and $S=2$, which is not true for either alone.  Thus, the interaction information in this case is equal to zero despite the presence of both redundant and synergistic interactions, because redundancy and synergy are balanced.

The situation is worse for four-variable interaction information, which is known to violate the interpretation that positive values indicate (pure) synergy and negative values indicate (pure) redundancy \cite{Bell2003,Anastassiou2007}.  To demonstrate, consider the case of 3-parity, which is the higher-order form of the exclusive-OR, or 2-parity, function mentioned earlier.  In this case, we have a system of four binary random variables, $S$ and ${\bf R} = \{R_1, R_2, R_3\}$, where the eight outcomes for ${\bf R}$ are equiprobable and $S = R_1 \oplus R_2 \oplus R_3$.  Intuitively, this corresponds to a case of pure synergy, since the value of $S$ can be determined only when all of the $R_i$ are known.  Indeed, using Eq. \eqref{nvarII} we find that $I(S; R_1; R_2; R_3)$ for this system is equal to $+1$ bit, as expected from the interpretation that positive values indicate synergy.  However, now consider a second system of binary variables, this time where the two outcomes of $S$ are equiprobable and $R_1$, $R_2$, and $R_3$ are all copies of $S$.  Clearly this corresponds to a case of pure redundancy, since the value of $S$ can be determined uniquely from knowledge of any $R_i$, but $I(S; R_1; R_2; R_3)$ for this system is again equal to $+1$ bit, same as the case of pure synergy.  Thus, a completely redundant system is assigned a positive value for the interaction information, in clear violation of the idea that redundancy is indicated by negative values.  Worse still, the 4-variable interaction information fails to distinguish between the polar opposites of purely synergistic and purely redundant information.  

\begin{figure}[t!]
  \centering
    \includegraphics{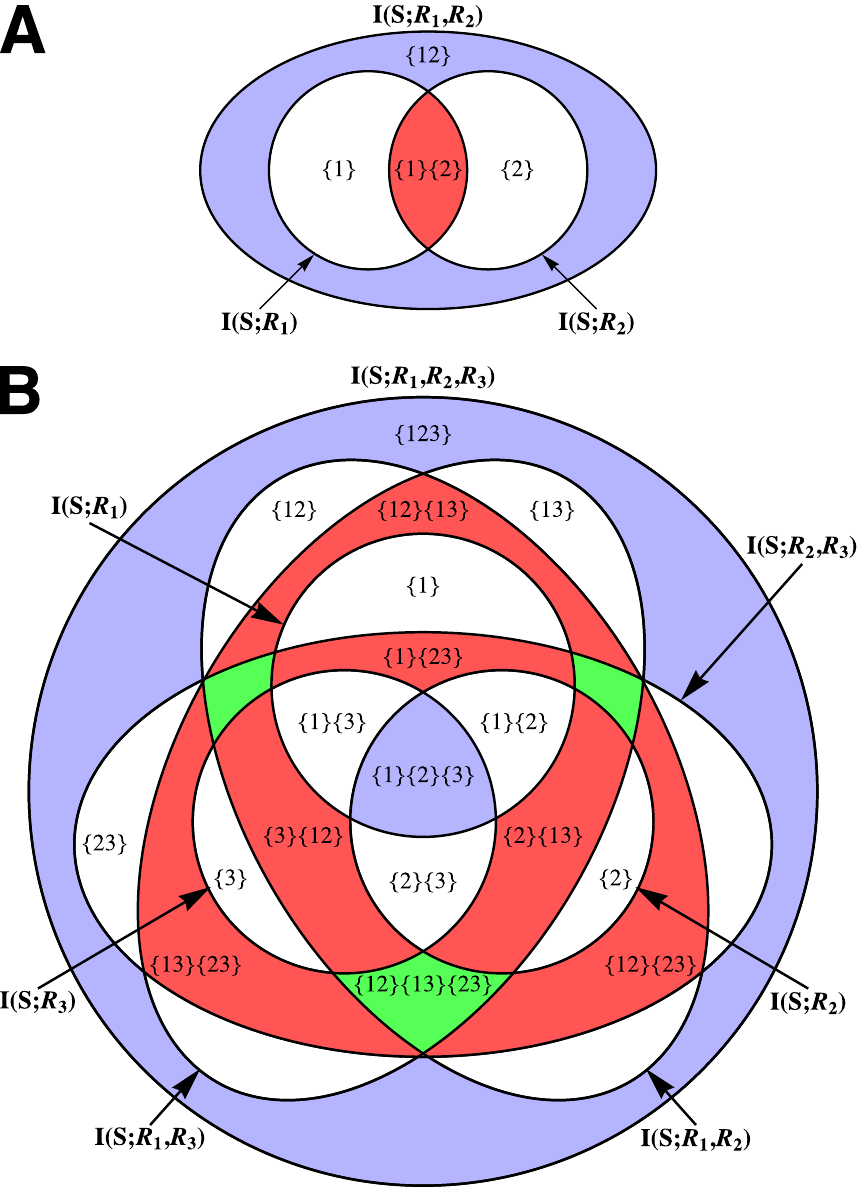}
\caption{PI-decomposition of interaction information for (A) 3 and (B) 4 variables.  Blue and red regions represent PI-terms that are added and subtracted, respectively.  The green region in (B) represents a PI-term that is subtracted twice.}
\label{interaction_info}
\end{figure}

The PI-decomposition for 4-variable interaction information (FIG. 5B; see also Fig. S4 in Appendix E) clarifies why this is the case.  In terms of PI-atoms, $I(S; R_1; R_2; R_3)$ is given by
\begin{align}
& \Pi_{{\bf R}}(S; \{123\}) + \Pi_{{\bf R}}(S; \{1\}\{2\}\{3\}) \notag \\
 - & \Pi_{{\bf R}}(S; \{1\}\{23\}) - \Pi_{{\bf R}}(S; \{2\}\{13\}) - \Pi_{{\bf R}}(S; \{3\}\{12\}) \notag \\
 - & \Pi_{{\bf R}}(S; \{12\}\{13\}) - \Pi_{{\bf R}}(S; \{12\}\{23\}) - \Pi_{{\bf R}}(S; \{13\}\{23\}) \notag \\
 - & 2 \times \Pi_{{\bf R}}(S; \{12\}\{13\}\{23\}).
\end{align}
Thus, $I(S; R_1; R_2; R_3)$ is equal to the sum of third-order synergy (\{123\}) and third-order redundancy (\{1\}\{2\}\{3\}), minus the information provided redundantly by a first- and second-order synergy (\{1\}\{23\}, \{2\}\{13\}, and \{3\}\{12\}), minus the information provided redundantly by two second-order synergies (\{12\}\{13\}, \{12\}\{23\}, and \{13\}\{23\}), and minus twice the information provided redundantly by all three second-order synergies (\{12\}\{13\}\{23\}).  Thus, systems with pure synergy and pure redundancy have the same value for $I(S; R_1; R_2; R_3)$ because 4-variable interaction information adds in the highest-order synergy and redundancy terms.  More generally, the PI-decomposition for $I(S; R_1; R_2; R_3)$ shows why it is difficult to interpret as a meaningful quantity, and as one might expect the story only becomes more complicated in higher dimensions.  Thus, although one can readily decompose interaction information into a collection of partial information contributions, and understand the conditions under which it will be positive or negative depending on the relative magnitudes of these contributions, the utility of interaction information for larger systems is unclear.  

\section{Discussion}

The main objective of this paper has been to quantify multivariate information in such a way that the structure of variable interactions is illuminated.  This was accomplished by first defining a general measure of redundant information, $I_{\min}$, which satisfies a number of intuitive properties for a measure of redundancy.  Next, it was shown that $I_{\min}$ induces a lattice structure over the set of possible information sources, referred to as the redundancy lattice, which characterizes the distinct ways that information can be distributed amongst a set of sources.  From this lattice, a measure of partial information was derived that captures the unique information contributed by each possible combination of sources.  It was then shown that mutual information decomposes into a sum of these partial information terms, so that the total information provided by a source is broken down into a collection of partial information contributions.  Moreover, it was demonstrated that each of these terms supports clear interpretation as a particular combination of redundant and synergistic interactions between specific subsets of variables.  Finally, we discussed the relationship between partial information decomposition and interaction information, the current de facto measure of multivariate interactions, and used partial information to clarify the confusing property that interaction information is sometimes negative.

One obvious challenge with applying these ideas is that the number of partial information terms grows rapidly for larger systems.  For instance, with 9 variables there are more than $5 \times 10^{22}$ possibilities \cite{Wiedemann1991}, and beyond that the Dedekind numbers are not even currently known.  Thus, clearly an important direction for future work is to determine efficient ways of calculating partial information terms for larger systems. To this end, the lattice structure of the terms is likely to play an essential role.  As with any ordered data structure, the fact that the space of possibilities is highly organized can be readily exploited for efficient use.  For instance, as a simple example, if $I_{\min}$ is calculated in a descending fashion over the nodes of the redundancy lattice and at a certain juncture has a value of zero, all of the terms below that node can immediately be eliminated simply from the monotonicity of $I_{\min}$ (see Appendix D).  Moreover, if the Markov property or any other constraints hold between the variables, many of the possible partial information terms can also be excluded.  Finally, these considerations notwithstanding, it should also be emphasized that 3-variable interaction is the current state of the art, and thus even the simplest form of partial information decomposition can be used to address a number of outstanding questions.

In physics, for example, 3-variable interactions have been explored in relation to the non-separability of quantum systems \cite{Cerf1997} and in the study of many-body correlation effects \cite{Matsuda2000}.  In neuroscience, the concepts of synergy and redundancy for three variables have been examined in the context of neural coding in a number of theoretical and empirical investigations \cite{Brenner2000,Panzeri1999,Schneidman2003,Latham2005,Gat1999,Narayanan2005}.  In genetics, multivariate dependencies arise in the analysis of gene-gene and gene-environment interactions in studies of human disease susceptibility \cite{Anastassiou2007,Moore2006, Chanda2007}.  Moreover, similar issues have also been explored in machine learning \cite{Pearl1988,Jakulin2003,Mackay2003}, ecology \cite{Orloci2002}, quantum information theory \cite{Vedral2002}, information geometry \cite{Amari2001}, rough set analysis \cite{Gediga2003}, and cooperative game theory \cite{Grabisch1999}.  Thus, in all of these cases, the 3-variable form of partial information decomposition can be applied immediately to illuminate the structure of multivariate dependencies, while the general form provides a clear way forward in the study of more complex systems of interactions.

\begin{acknowledgments}
We thank O. Sporns, J. Beggs, A. Kolchinsky, and L. Yaeger for helpful comments.  This work was supported in part by NSF grant IIS-0916409 (to R.D.B.) and an NSF IGERT traineeship (to  P.L.W.).
\end{acknowledgments}

\appendix

\section{Measures of Specific Information}

Measures of specific information are discussed in \cite{Deweese1999} in the context of quantifying the information that specific neural responses provide about a stimulus ensemble.  For random variables $S$ and $R$, representing stimuli and responses, respectively, the information that $R$ provides about $S$ is decomposed according to
\begin{equation}
I(S; R) = \sum_{r \in R} p(r) i_{r}(r) \tag{A1}
\end{equation}
and
\begin{equation}
i_{r}(r) = H(S) - H(S | r) \tag{A2}
\end{equation}
where $H(S)$ is the entropy of $S$ and $i_{r}(r)$ is the \emph{response-specific information} associated with each $r \in R$.  The response-specific information quantifies the change in uncertainty about $S$ when response $r$ is observed.  In \cite{Deweese1999}, it is shown that $i_{r}$ is the unique measure of specific information that satisfies additivity, though it is also possible for $i_{r}$ to be negative.

To distinguish the different role played by stimuli as opposed to responses, an alternative measure of specific information is proposed in \cite{Butts2003}.  The \emph{stimulus-specific information} for an outcome $s \in S$ is defined as
\begin{equation*}
i_{s}(s) = \sum_{r \in R} p(r|s) i_r(r). \tag{A3}
\end{equation*}
Like the response-specific information, the weighted average of $i_{s}(s)$ gives the mutual information $I(S;R)$.  Stimulus-specific information quantifies the extent to which a particular stimulus $s$ tends to evoke responses that are informative about the entire ensemble $S$ (responses with high values for $i_{r}$).

Finally, both \cite{Deweese1999} and \cite{Butts2003} also discuss $I(S=s; R)$, the measure of specific information used here (Eq. \eqref{specInfo}).  In \cite{Butts2003}, $I(S=s; R)$ is described as the reduction in surprise of a particular stimulus $s$ gained from each response, averaged over all responses associated with that stimulus.  Thus, whereas $i_{s}(s)$ weights each response $r$ according to the information that it contributes about the entire ensemble $S$, $I(S=s; R)$ quantifies only the information that $R$ provides about the particular outcome $S=s$.  In \cite{Deweese1999}, it is proven that $I(S=s; R)$ is the only measure of specific information that is strictly nonnegative.

\section{Lattice Theory Definitions}

\begin{figure*}[t!]
  \centering
    \includegraphics[width=0.75\textwidth]{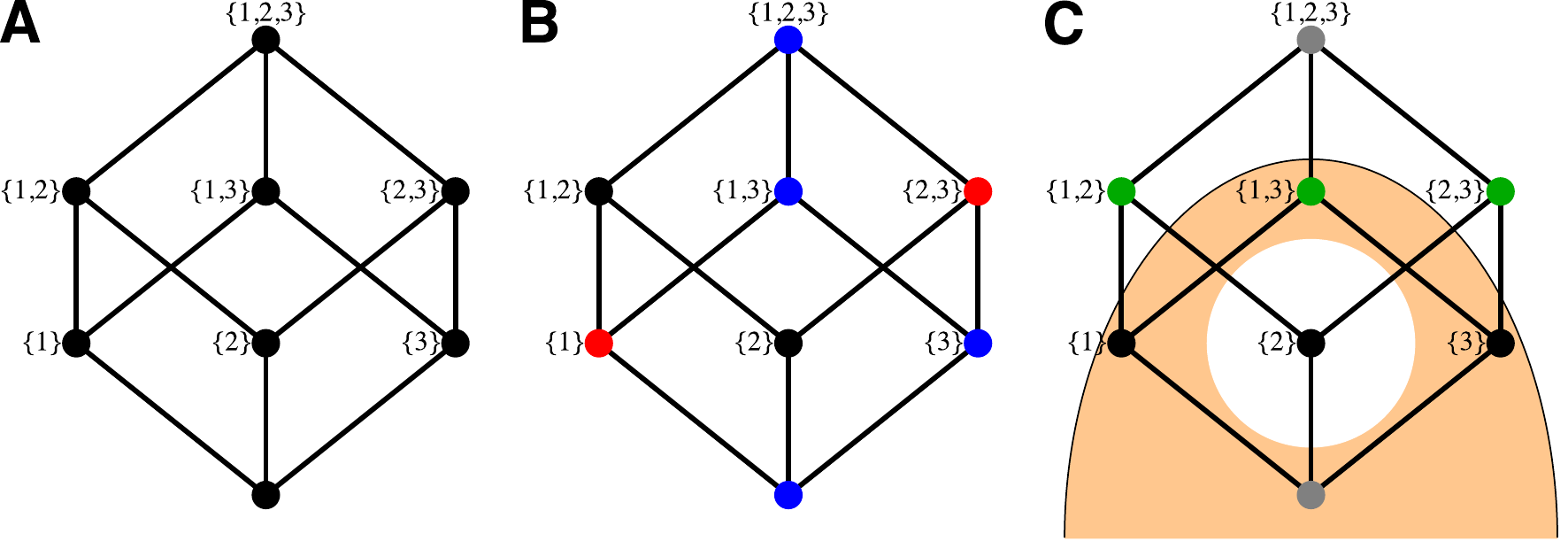}
       \begin{flushleft}FIG. S1: Basic lattice-theoretic concepts. (A) Hasse diagram of the lattice $\langle \mathcal{P}(X) , \subseteq\rangle$ for $X = \{1, 2, 3\}$. (B) An example of a chain (blue nodes) and an antichain (red nodes). (C) The top $\top$ and bottom $\bot$ are shown in gray.  Green nodes correspond to $\{1, 2, 3\}^-$, the set of elements covered by $\{1, 2, 3\}$. The orange region represents $\downarrow \{1,3\}$, the down-set of $\{1, 3\}$.\end{flushleft}
\end{figure*}

Here we review only the basic concepts of lattice theory needed for supporting proofs.  For a thorough treatment, see \cite{Davey2002, Gratzer2003}.

\begin{mydef}
A pair $\langle X, \leqslant \rangle$ is a \emph{partially ordered set} or \emph{poset} if  $\leqslant$ is a binary relation on $X$ that is reflexive, transitive and antisymmetric.
\end{mydef}

\begin{mydef} \label{maxMinSets}
Let $Y \subseteq X$.  Then $a \in Y$ is a \emph{maximal element} in $Y$ if for all $ b \in Y, a \leqslant b \Rightarrow a = b$.  A \emph{minimal element} is defined dually.  We denote the set of maximal elements of $Y$ by $\overline{Y}$ and the set of minimal elements by $\underline{Y}$.
\end{mydef}

\begin{mydef}
Let $\langle X, \leqslant \rangle$ be a poset, and let $Y \subseteq X$.  An element $x \in X$ is an \emph{upper bound} for $Y$ if for all $y \in Y, y \leqslant x$.  A \emph{lower bound} for $Y$ is defined dually.
\end{mydef}

\begin{mydef}
An element $x \in X$ is the \emph{least upper bound} or \emph{supremum} for $Y$, denoted $\sup Y$, if $x$ is an upper bound of $Y$ and for all $y \in Y$ and all $z \in X, y \leqslant z$ implies $x \leqslant z$.  The \emph{greatest upper bound} or \emph{infimum} for $Y$, denoted $\inf Y$, is defined dually.
\end{mydef}

\begin{mydef}
A poset $\langle X, \leqslant \rangle$ is a \emph{lattice} if, and only if, for all $x, y \in X$ both $\inf \{x, y\}$ and $\sup \{x, y\}$ exist in $X$.  If $\langle X, \leqslant \rangle$ is a lattice, it is common to write $x \wedge y$, the \emph{meet} of $x$ and $y$, and $x \vee y$, the \emph{join} of $x$ and $y$, for $\inf \{x, y\}$ and $\sup \{x, y\}$, respectively.  For $Y \subseteq X$, we use $\bigwedge Y$ and $\bigvee Y$ to denote the meet and join of all elements in $Y$, respectively.
\end{mydef}

\begin{mydef}
For $a, b \in X$, we say that $a$ is \emph{covered by} $b$ (or $b$ \emph{covers} $a$) if $a < b$ and $a \leqslant c < b \Rightarrow a = c$.  The set of elements that are covered by $b$ is denoted by $b^-$.
\end{mydef}

The classic example of a lattice is the power set of a set $X$ ordered by inclusion, denoted $\langle \mathcal{P}(X), \subseteq\rangle$.  Lattices are naturally represented by Hasse diagrams, in which nodes correspond to members of $X$ and an edge exists between elements $x$ and $y$ if $x$ covers $y$.  FIG. S1A depicts the Hasse diagram for the lattice $\langle \mathcal{P}(X) , \subseteq\rangle$ with $X = \{1, 2, 3\}$.

\begin{mydef}
If $\langle X, \leqslant \rangle$ is a poset, $Y \subseteq X$ is a \emph{chain} if for all $a, b \in Y$ either $a \leqslant b$ or $b \leqslant a$.  Y is an \emph{antichain} if $a \leqslant b$ only if $a = b$. 
\end{mydef}

FIG. S1B shows examples of a chain and an antichain.

\begin{mydef}
If there exists an element $\bot \in X$ with the property that $\bot \leqslant x$ for all $x \in X$, we call $\bot$ the \emph{bottom element} of $X$.  The \emph{top element} of X, denoted by $\top$, is defined dually.
\end{mydef}

\begin{mydef}
For any $x \in X$, we define
\begin{equation*}
\downarrow x = \{y \in X : y \leqslant x\} \mbox{ and } \dot{\downarrow} x = \{y \in X: y < x\}
\end{equation*}
where $\downarrow x$ and $\dot{\downarrow} x$ are called the \emph{down-set} and \emph{strict down-set} of $x$, respectively.
\end{mydef}

FIG. S1C illustrates the concepts of top and bottom elements, covering relations, and down-sets.

\section{$\mathcal{A}({\bf R})$ and the Redundancy Lattice}

Formally, $\mathcal{A}({\bf R})$ corresponds to the set of antichains on the lattice $\langle\mathcal P({\bf R}),\subseteq\rangle$ (excluding the empty set).  The cardinality of this set for $|{\bf R}|=n-1$ is given by the $(n-1)$-th Dedekind number, which for $n = 2, 3, 4,\ldots$ is $1, 4, 18, 166, 7579, \ldots$ (\cite{Comtet1974}, p. 273).  The fact that $\langle \mathcal{A}({\bf R}), \preccurlyeq \rangle$ forms a lattice, which we call the redundancy lattice, is proven in \cite{Crampton2000}, where the corresponding lattice is denoted $\langle \mathcal{A}(X), \preccurlyeq^\prime \rangle$ (see also \cite{Crampton2001}).  As shown in \cite{Crampton2000}, the meet ($\wedge$) and join ($\vee$) for this lattice are given by
\begin{align}
\alpha \wedge \beta & = \underline{\alpha \cup \beta} \label{RLmeet}\tag{A4}\\ 
& \mbox{and} \notag\\
\alpha \vee \beta & = \underline{\uparrow\alpha \cap \uparrow\beta} \tag{A5}.
\end{align}

\section{Supporting Proofs}

\begin{mythm}\label{SpecInfoNonNeg} 
$I(S=s; {\bf A})$ is nonnegative.
\end{mythm}
\begin{proof}
\begin{align*}
I(S=s; {\bf A}) = D(p({\bf a}|s) \parallel p({\bf a})) \geq 0
\end{align*}
where $D$ is the Kullback-Leibler distance and the last step follows from the information inequality (\cite{Cover2006}, p. 26).
\end{proof}

\begin{mylma}\label{SpecInfoIncr} 
$I(S=s; {\bf A})$ increases monotonically on the lattice $\langle \mathcal{P}({\bf R}) , \subseteq\rangle$.
\end{mylma}
\begin{widetext}
\begin{proof}
Consider  ${\bf A}, {\bf B}$ with ${\bf A} \subset {\bf B} \subseteq {\bf R}$. Let ${\bf C} = {\bf B} \setminus {\bf A} \neq \emptyset$. Then we have
\begin{align*}
& I(S=s; {\bf B}) - I(S=s; {\bf A}) \\
= &  \sum_{{\bf b}} p({\bf b}|s) \log \frac{p(s, {\bf b})}{p(s) p({\bf b})} - \sum_{{\bf a}} p({\bf a}|s) \log \frac{p(s, {\bf a})}{p(s) p({\bf a})} \\ 
= &  \sum_{{\bf a}} \sum_{{\bf c}} p({\bf a}, {\bf c}|s) \log \frac{p(s, {\bf a}, {\bf c})}{p(s) p({\bf a}, {\bf c})} - \sum_{{\bf a}} \sum_{{\bf c}} p({\bf a}, {\bf c}|s) \log \frac{p(s, {\bf a})}{p(s) p({\bf a})} \\
= & \sum_{{\bf a}} \sum_{{\bf c}} p({\bf a}, {\bf c}|s) \log \frac{p(s, {\bf c}|{\bf a})}{p(s|{\bf a}) p({\bf c}|{\bf a})} \\
= & \sum_{{\bf a}} p({\bf a}) \sum_{{\bf c}} p({\bf c} | {\bf a}, s) \log \frac{p({\bf c} | {\bf a}, s)}{p({\bf c}|{\bf a})} \\
= & \sum_{{\bf a}} p({\bf a}) D(p({\bf c} | {\bf a}, s)\parallel p({\bf c}|{\bf a})) \geq 0.
\end{align*}
\end{proof}
\end{widetext}

\begin{mythm}\label{IminIncr}
$I_{\min}$ increases monotonically on the lattice $\langle \mathcal{A}({\bf R}), \preccurlyeq\rangle$.
\end{mythm}
\begin{proof}
We proceed by contradiction.  Assume there exists $\alpha, \beta \in \mathcal{A}({\bf R})$ with $\alpha \prec \beta$ and $I_{\min}(S; \beta) < I_{\min}(S; \alpha)$.  Then, from Eq. \eqref{IminDef}, there must exist ${\bf B} \in \beta$ such that $I(S=s; {\bf B}) < I(S=s; {\bf A})$ for some outcome $s \in S$ and for all ${\bf A} \in \alpha$. Thus, from Lemma \ref{SpecInfoIncr}, there does not exist ${\bf A} \in \alpha$ such that ${\bf A} \subseteq {\bf B}$.  However, since $\alpha \prec \beta$ by assumption, there exists ${\bf A} \in \alpha$ such that ${\bf A} \subseteq {\bf B}$.
\end{proof}

\begin{mythm}\label{piCF1}
$\Pi_{{\bf R}}$ can be stated in closed form as
\begin{equation}
\Pi_{{\bf R}}(S; \alpha) = I_{\min}(S; \alpha) - \sum_{k=1}^{|\alpha^-|} (-1)^{k-1} \sum_{\substack{\mathcal{B} \subseteq \alpha^- \\|\mathcal{B}|=k}} I_{\min}(S; \bigwedge \mathcal{B}). \tag{A6} \label{piCF1Eq}
\end{equation}
\end{mythm}
\begin{proof} 
For $\mathcal{B} \subseteq \mathcal{A}({\bf R})$, define the set-additive function $f$ as
\begin{equation*}
f(\mathcal{B}) = \sum_{\beta \in \mathcal{B}} \Pi_{{\bf R}}(S; \beta).
\end{equation*}
From Eq. \eqref{piImplicit}, it follows that $I_{\min}(S; \alpha) = f(\downarrow \alpha)$ and
\begin{align*}
\Pi_{{\bf R}}(S; \alpha) & = f(\downarrow \alpha) - f(\dot{\downarrow} \alpha) \\ 
& = f(\downarrow \alpha) - f(\bigcup_{\beta \in \alpha^-} \downarrow \beta).
\intertext{Applying the principle of inclusion-exclusion (\cite{Stanley1997}, p. 64), we have}
& = f(\downarrow \alpha) - \sum_{k=1}^{|\alpha^-|} (-1)^{k-1} \sum_{\substack{\mathcal{B} \subseteq \alpha^- \\|\mathcal{B}|=k}} f(\bigcap_{\gamma \in \mathcal{B}} \downarrow \gamma)
\intertext{and it is a basic result of lattice theory that for any lattice $L$ and $A \subseteq L$, $\bigcap_{a \in A} \downarrow a = \downarrow (\bigwedge A)$ (\cite{Davey2002}, p. 57), so we have}
& = f(\downarrow \alpha) - \sum_{k=1}^{|\alpha^-|} (-1)^{k-1} \sum_{\substack{\mathcal{B} \subseteq \alpha^- \\|\mathcal{B}|=k}} f(\downarrow (\bigwedge \mathcal{B})) \\
& = I_{\min}(S; \alpha) - \sum_{k=1}^{|\alpha^-|} (-1)^{k-1} \sum_{\substack{\mathcal{B} \subseteq \alpha^- \\|\mathcal{B}|=k}} I_{\min}(S; \bigwedge \mathcal{B}).
\end{align*}
\end{proof}

\begin{mylma}[Maximum-minimums identity]\label{MaxMinsIdent} 
Let $A$ be a set of numbers.  The maximum-minimums identity states that
\begin{align*}
\max A & = \sum_{k=1}^{|A|} (-1)^{k-1} \sum_{\substack{B\subseteq A\\|B|=k}} \min B 
\shortintertext{or conversely,}
\min A & = \sum_{k=1}^{|A|} (-1)^{k-1} \sum_{\substack{B\subseteq A\\|B|=k}} \max B.
\end{align*}
\end{mylma}
\begin{proof}
It is proven in a number of introductory texts, e.g. \cite{Ross2009}.
\end{proof}

\begin{widetext}
\begin{mythm} \label{piDefAlt}
$\Pi_{{\bf R}}$ can be stated in closed form as
\begin{equation*}
\Pi_{{\bf R}}(S; \alpha) = I_{\min}(S; \alpha) -  \sum_{s} p(s) \max_{\beta \in \alpha^-} \min_{{\bf B} \in \beta} I(S=s; {\bf B}). \tag{A7}
\end{equation*}
\end{mythm}
\begin{proof}
\begin{align*}
\intertext{Combining Eqs. \eqref{piCF1Eq} and \eqref{IminDef} yields}
\Pi_{{\bf R}}(S; \alpha) & = I_{\min}(S; \alpha) - \sum_{k=1}^{|\alpha^-|} (-1)^{k-1} \sum_{\substack{\mathcal{B} \subseteq \alpha^- \\|\mathcal{B}|=k}} \sum_{s} p(s) \min_{{\bf B} \in \bigwedge \mathcal{B}} I(S=s; {\bf B})\\
& = I_{\min}(S; \alpha) - \sum_{s} p(s) \sum_{k=1}^{|\alpha^-|} (-1)^{k-1} \sum_{\substack{\mathcal{B} \subseteq \alpha^- \\|\mathcal{B}|=k}} \min_{{\bf B} \in \bigwedge \mathcal{B}} I(S=s; {\bf B})
\intertext{and by Lemma \ref{SpecInfoIncr} and Eq.~(\ref{RLmeet}),}
& = I_{\min}(S; \alpha) - \sum_{s} p(s) \sum_{k=1}^{|\alpha^-|} (-1)^{k-1} \sum_{\substack{\mathcal{B} \subseteq \alpha^- \\|\mathcal{B}|=k}} \min_{\beta \in \mathcal{B}} \min_{{\bf B} \in \beta} I(S=s; {\bf B}).
\intertext{Then, applying Lemma \ref{MaxMinsIdent} we have}
& = I_{\min}(S; \alpha) - \sum_{s} p(s) \max_{\beta \in \alpha^-} \min_{{\bf B} \in \beta} I(S=s; {\bf B}).
\end{align*}
\end{proof}

\begin{mythm}\label{piNonNeg}
$\Pi_{{\bf R}}$ is nonnegative.
\end{mythm}
\begin{proof}
If $\alpha = \bot$, $\Pi_{{\bf R}}(S; \alpha) = I_{\min}(S; \alpha)$ and $\Pi_{{\bf R}}(S; \alpha) \geq 0$ follows from the nonnegativity of $I_{\min}$.   To prove it for $\alpha \neq \bot$, we proceed by contradiction.  Assume there exists $\alpha \in \mathcal{A}({\bf R}) \setminus \{\bot\}$ such that $\Pi_{{\bf R}}(S; \alpha) < 0$.  Applying Eq. \eqref{IminDef} to Theorem \ref{piDefAlt} and combining summations yields
\begin{equation*}
\Pi_{{\bf R}}(S; \alpha) = \sum_{s} p(s) \{ \min_{{\bf A} \in \alpha} I(S=s; {\bf A}) -  \max_{\beta \in \alpha^-} \min_{{\bf B} \in \beta} I(S=s; {\bf B}) \}.
\end{equation*}
From this equation, it is clear that there must exist $\beta \in \alpha^-$ such that for all ${\bf B} \in \beta$, $I(S=s; {\bf A}) < I(S=s; {\bf B})$ for some outcome $s \in S$ and some ${\bf A} \in \alpha$.  Thus, from Lemma \ref{SpecInfoIncr}, there does not exist ${\bf B} \in \beta$ such that ${\bf B} \subseteq {\bf A}$.  However, since $\beta \prec \alpha$ by definition, there exists ${\bf B} \in \beta$ such that ${\bf B} \subseteq {\bf A}$.
\end{proof}

\newpage

\section{Supplementary Figures}
\begin{figure*}[h!]
  \centering
    \includegraphics[width=.9\textwidth]{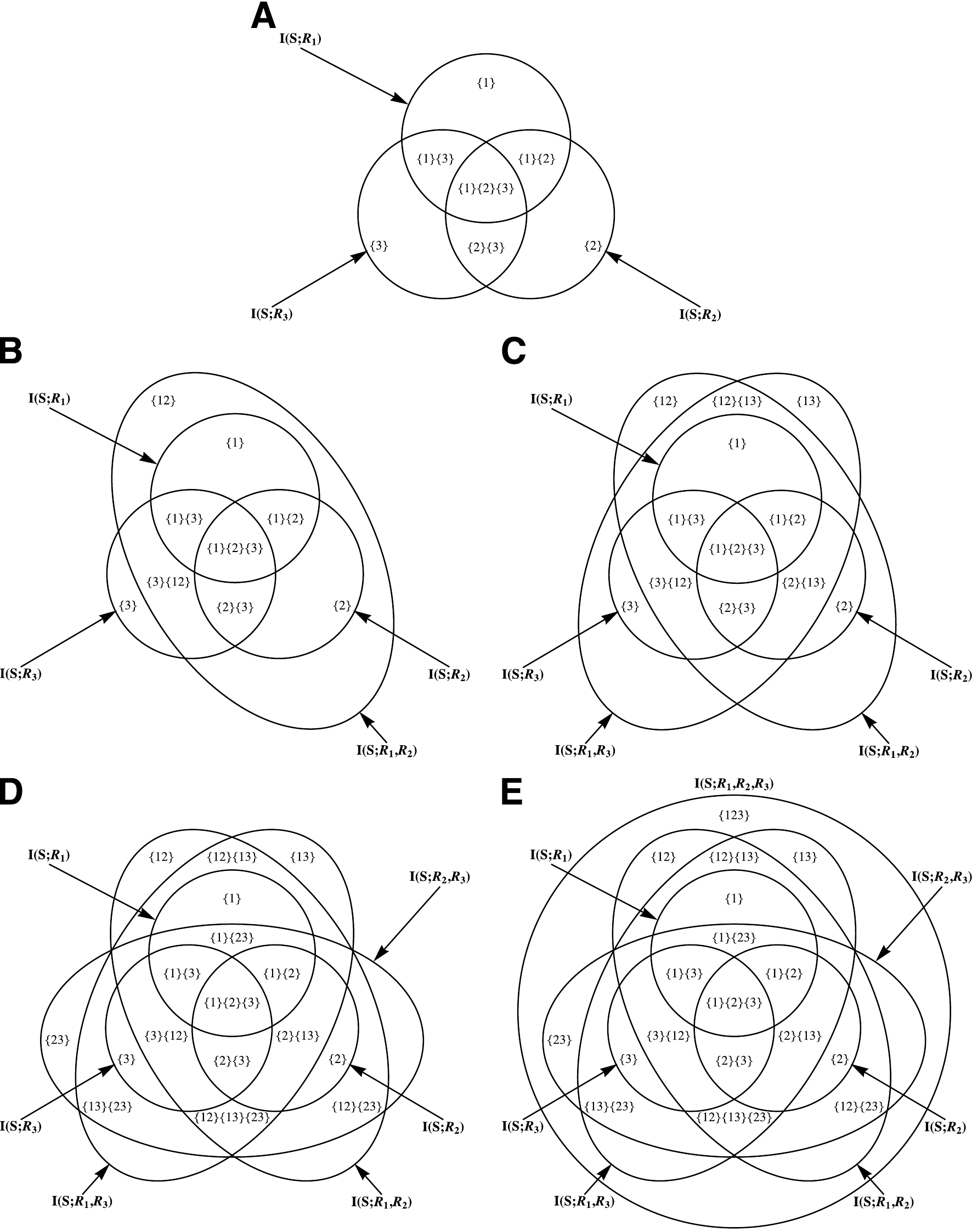}
   \begin{flushleft}FIG. S2: Constructing a PI-diagram for 4 variables. (A) For each element $R_i \in {\bf R}$  there is a region corresponding to $I(S; R_i)$.  (B-E) For each subset ${\bf A}$ of ${\bf R}$ with two or more elements, $I(S; {\bf A})$ is depicted as a region containing $I(S; A)$ for all $A \in {\bf A}$ but not coextensive with $\bigcup_{A \in {\bf A}} I(S; A)$.  Regions of the diagram intersect generically, representing all possibilities for redundancy.\end{flushleft}
\end{figure*}

\begin{figure*}[p!]
  \centering
    \includegraphics[width=0.65\textwidth]{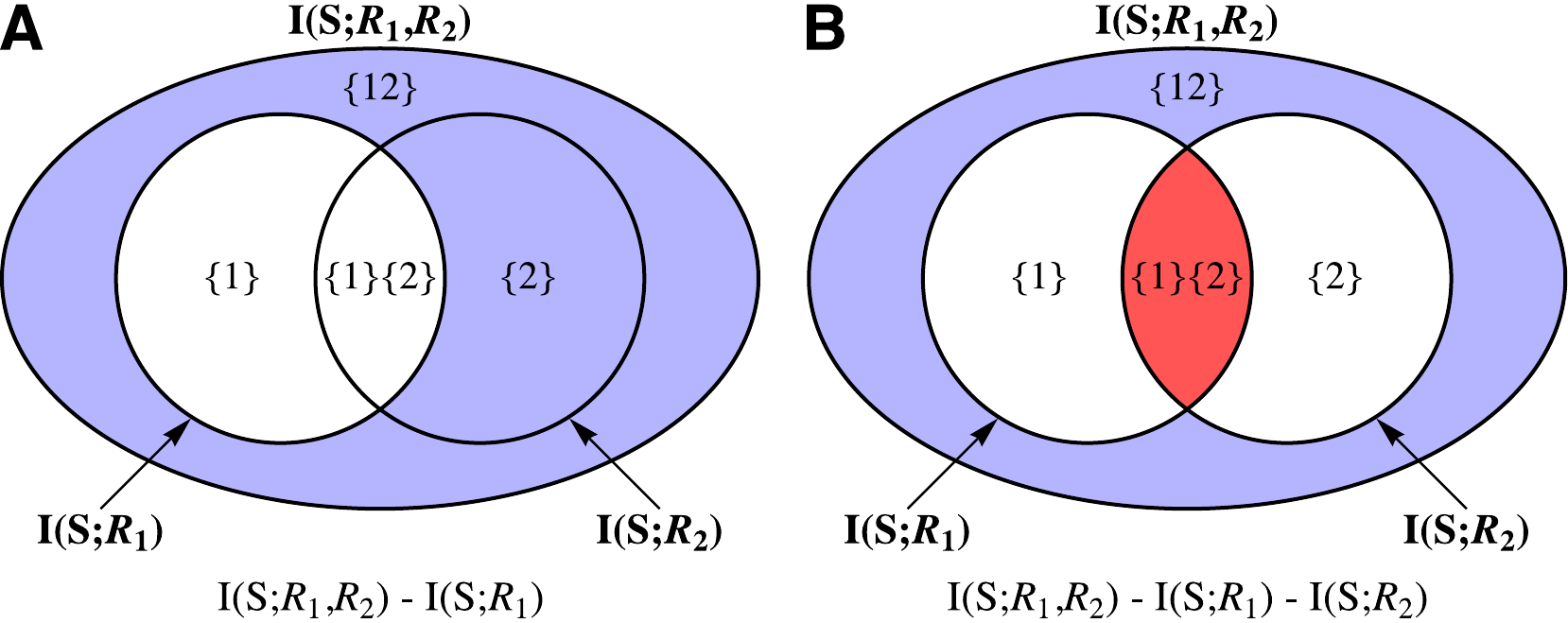}
       \begin{flushleft}FIG. S3: Computing the PI-decomposition for 3-variable interaction information. (A-B) Term-by-term calculation of $I(S;R_1;R_2) = I(S;R_1,R_2) - I(S;R_1) - I(S;R_2)$. Blue and red regions represent PI-terms that are added and subtracted, respectively.\end{flushleft}
\end{figure*}

\begin{figure*}[p!]
  \centering
    \includegraphics[width=0.82\textwidth]{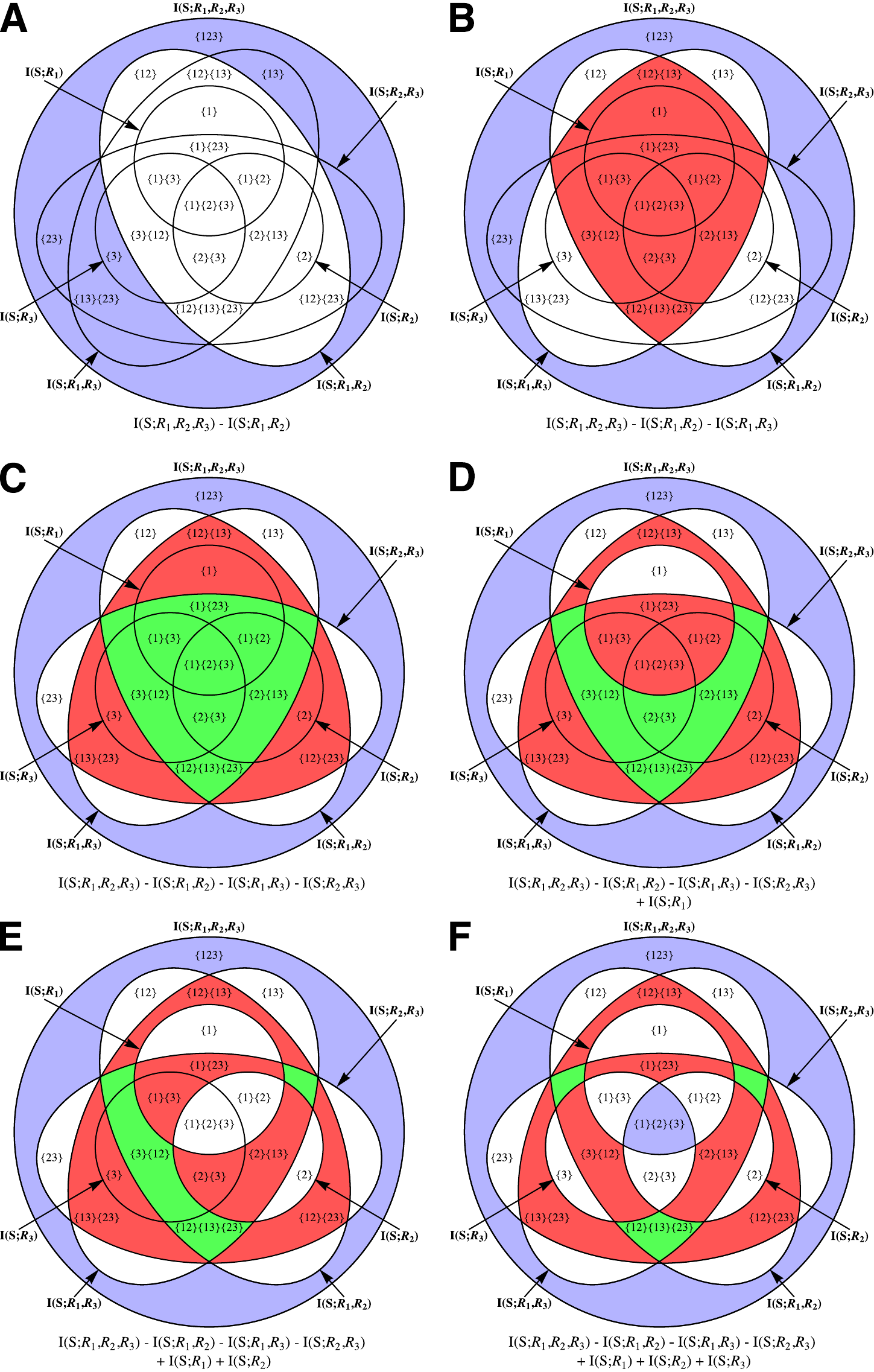}
           \begin{flushleft}FIG. S4: Computing the PI-decomposition for 4-variable interaction information.  (A-F) Term-by-term calculation of $I(S;R_1;R_2;R_3) = I(S;R_1,R_2,R_3) - I(S;R_1,R_2) - I(S;R_1,R_3) - I(S;R_2,R_3) + I(S;R_1) + I(S;R_2) + I(S;R_3)$.  Blue and red regions represent PI-terms that are added and subtracted, respectively.  Green regions represent PI-terms that are subtracted twice.\end{flushleft}
\end{figure*}

\end{widetext}

\end{document}